\documentclass[submission,copyright,creativecommons]{eptcs}
\providecommand{\event}{AFL 2017} 
\usepackage{breakurl}             
\usepackage{underscore}           
\usepackage{cite}
\usepackage{amssymb}
\usepackage{amsmath}
\usepackage{amsthm}
\usepackage{subfigure}
\usepackage{adjustbox}
\usepackage{url}
\usepackage{tikz}
\usepackage{tikzscale}
\usepackage{pgflibraryshapes}
\usetikzlibrary{arrows,positioning,automata}

\newtheorem{theorem}{Theorem}
\newtheorem{corollary}{Corollary}[theorem]
\newtheorem{lemma}{Lemma}
\newtheorem{definition}{Definition}
\newtheorem{example}{Example}

\title{Descriptional Complexity of Non-Unary Self-Verifying Symmetric Difference Automata}
\author{Laurette Marais
\institute{Department of Computer Science, Stellenbosch University, South Africa}
\institute{Meraka Institute, CSIR, South Africa}
\and
Lynette van Zijl
\institute{Department of Computer Science, Stellenbosch University, South Africa}}

\def\titlerunning{Descriptional Complexity of Non-Unary SV-XNFA}
\def\authorrunning{L. Marais \& L. Van Zijl}
\begin{document}
\maketitle

\begin{abstract}
Previously, self-verifying symmetric difference automata were defined and a tight bound of $2^{n-1}-1$ was shown for state complexity in the unary case. We now consider the non-unary case and show that, for every $n\geq 2$, there is a regular language $\mathcal{L}_{n}$ accepted by a non-unary self-verifying symmetric difference nondeterministic automaton with $n$ states, such that its equivalent minimal deterministic finite automaton has $2^{n-1}$ states. Also, given any SV-XNFA with $n$ states, it is possible, up to isomorphism, to find at most another $|GL(n,\mathbb{Z}_2)|-1$ equivalent SV-XNFA.
\end{abstract}
\section{Introduction}
\label{sec:intro}
Symmetric difference nondeterministic finite automata (XNFA) are interesting from a state complexity point of view. Determinising XNFA is done via the subset construction as for NFA, but instead of taking the union of sets, the symmetric difference is taken. This means that $2^n-1$ is an upper bound on the state complexity of XNFA. This has been shown to be a tight bound for unary alphabets\cite{VanZijl2004}.

Self-verifying automata (SV-NFA) were described in  \cite{Assent2007,Hromkovic1999,Jiraskova2011} as having two kinds of final states: accept states and reject states. Non-final states are called neutral states. It is required that for any word, at least one such a final state is reached, and that only one kind of final state is reached on any path, so that any word is either explicitly accepted or explicitly rejected by the automaton. It was shown in \cite{Jiraskova2011} that $e^{\Theta\sqrt{n\ \ln n}}$ is an upper bound for the unary case, but not a tight bound, while in the non-unary case, $g(n)$, where $g(n)$ grows like $3^{\frac{n}{3}}$, is a tight upper bound.

In \cite{Marais2016}, we extended the notion of self-verification (SV) to XNFA to obtain SV-XNFA. We showed that $2^n-1$ is not a tight upper bound for SV-XNFA in the case of a unary alphabet. A lower bound of $2^{n-1}-1$ was established for the unary case, and we showed this to be a tight bound in \cite{Marais2017}.

In this paper, we now consider the state complexity of SV-XNFA with \textit{non-unary} alphabets. We give an upper bound of $2^n-1$ and a lower bound of $2^{n-1}$.

Furthermore, any XNFA can be transformed into an equivalent XNFA by performing a change of basis operation \cite{VanderMerwe2012}. We show that this holds also for SV-XNFA, and that for any given SV-XNFA, up to isomorphism, at most another $|GL(n,\mathbb{Z}_2)|-1$ equivalent SV-XNFA can be found.

\section{Preliminaries} 
\label{sec:prelim}
An NFA $N$ is a five-tuple $N=(Q,\Sigma,\delta,Q_0,F)$, where $Q$ is a finite set of states, $\Sigma$ is a finite alphabet, $\delta:Q\times\Sigma\rightarrow 2^Q$ is a transition function (where $2^Q$ indicates the power set of $Q$), $Q_0\subseteq Q$ is a set of initial states, and $F\subseteq Q$ is the set of final, or acceptance, states. The transition function $\delta$ can be extended to strings in the Kleene closure $\Sigma^*$ of the alphabet. Let $w = \sigma_0\sigma_1\ldots \sigma_k$, then
\begin{displaymath}
\delta'(q,w) = \delta'(q,\sigma_0\sigma_1\cdots \sigma_k)=\delta(\delta(\cdots\delta(q,\sigma_0),\sigma_1),\ldots,\sigma_k)\ .
\end{displaymath}
For convenience, we write $\delta(q,w)$ to mean $\delta'(q,w)$.

An NFA $N$ is said to accept a string $w\in\Sigma^*$ if $q_0\in Q_0$ and $\delta(q_0,w)\in F$, and the set of all strings (also called words) accepted by $N$ is the language $\mathcal{L}(N)$ accepted by $N$.  Any NFA has an equivalent DFA which accepts the same language.  The DFA $N_D = (Q_D, \Sigma,\delta_D,Q_{0D},F_D)$ that is equivalent to a given NFA is found by performing the subset construction~\cite{Hopcroft1990}. In essence, the subset construction keeps track of all the states that the NFA may be in at the same time, and forms the states of the equivalent DFA by a grouping of the states of the DFA.  In short,
\begin{displaymath}
\delta_D(A,\sigma)=\bigcup_{q\in A}\delta(q,\sigma)
\end{displaymath}
for any $A\subseteq Q$ and $\sigma\in\Sigma$. Any $A$ is a final state in the DFA if $A \cap F \neq \emptyset$.

\subsection{Symmetric difference automata (XNFA)}
\label{subsec:xnfa}
A symmetric difference NFA (XNFA) is defined similarly to an NFA (including the extended transition function over $\Sigma^*$ as for NFA), except that the DFA equivalent to the XNFA is found by taking the symmetric difference (in the set theoretic sense) in the subset construction. That is, for any two sets $A$ and $B$, the symmetric difference is given by $\oplus(A,B)=(A\cup B)\setminus (A\cap B)$.   The subset construction is then applied as  
\begin{displaymath}
\label{eq:subset:star}
\delta_D(A,\sigma)=\bigoplus_{q\in A} \delta(q,\sigma)
\end{displaymath}
for any $A\subseteq Q$ and $\sigma\in\Sigma$.

For clarity, the DFA equivalent to an XNFA $N$ is termed an XDFA and denoted with $N_D$, where $N_D = (Q_D,\delta_D,Q_{0D},F_D)$. Note that $\delta_D: 2^Q \times \Sigma \rightarrow 2^Q$. It is customary to require that an XDFA final state consist of an odd number of final XNFA states, as an analogy to the symmetric difference set operation~\cite{Vuillemin2009} -- this is known as parity acceptance. XNFA accept the class of regular languages \cite{Vuillemin2009}.

Given parity acceptance, XNFA have been shown to be equivalent to weighted automata over the finite field of two elements, or GF(2) \cite{VanderMerwe2012, Vuillemin2009}. For an XNFA $N = (Q,\Sigma,\delta,Q_0,F)$, the transitions for each alphabet symbol $\sigma$ can be represented as a matrix over GF(2). Each row represents a mapping from a state $q\in Q$ to a set of states $P\in 2^Q$. $P$ is written as a vector with a one in position $i$ if $q_i\in P$, and a zero in position $i$ if $q_i\not\in P$. Hence, the transition table is represented as a matrix $M_\sigma$ of zeroes and ones (see Example~\ref{ex:binary-xnfa-gf2n}). This is known as the characteristic or transition matrix for $\sigma$ of the XNFA. In the rest of this paper, we consider only SV-XNFA with non-singular matrices, whose cycle structures do not include transient heads, i.e. states that are only reached once before a cycle is reached.

Initial and final states are similarly represented by vectors, and appropriate vector and matrix multiplications over GF(2) represent the behaviour of the XNFA\footnote{In GF(2), $1+1=0$.}. For instance, in the unary case we would have a single matrix $M_a$ that describes the transitions on $a$ for some XNFA with $n$ states. We encode the initial states $Q_0$ as vector of length $n$ over GF(2), namely $v(Q_0) = [q_{0_0} \; q_{0_1} \; \cdots \; q_{0_{n-1}}]$, where $q_{0_i} = 1$ if $q_i \in Q_0$ and $0$ otherwise. Similarly, we encode the final states as a length $n$ vector $v(F) = [q_{F_0} \; q_{F_1} \; \cdots \; q_{F_{n-1}}]$. Then $v(Q_0)M_a$ is a vector that encodes the states reached after reading the symbol $a$ exactly once, and $v(Q_0)M_a^k$ encodes the states reached after reading the symbol $a$ $k$ times. The weight of a word $w_k$ of length $k$ is given by
$$\Delta(w_k) = v(Q_0)M_a^kv(F)^T \ .$$
We can say that $M_a$ represents the word $a$, and $M_{a^k} = M_a^k$ represents the word $a^k$. In the binary case, we would have two matrices, $M_a$ for transitions on $a$ and $M_b$ for transitions on $b$. Reading an $a$ corresponds to multiplying by $M_a$, while reading a $b$ corresponds to multiplying by $M_b$. Let $M_w$ be the result of the appropriate multiplications of $M_a$ and $M_b$ representing some $w \in \{a,b\}^*$, then the weight of $w$ is given by $\Delta(w) = v(Q_0)M_wv(F)^T$.

We now show that, in the unary case, a so-called change of basis is possible, where for some $n \times n$ transition matrix $M_a$ of an XNFA and any non-singular $n \times n$ matrix $A$, $M'_a = A^{-1}M_aA$ is the transition matrix of an equivalent XNFA with $v(Q'_0) = v(Q_0)A$ and $v(F')^T = A^{-1}v(F)^T$. For any word $w_k$ of length $k$, we have the following:
\begin{align*}
\Delta'(w_k) &= v(Q'_0)M'^k_a v(F')^T \\
             &= v(Q_0)A (A^{-1}M_aA)^k A^{-1}v(F)^T \\
             &= v(Q_0)M_a^k v(F)^T \\
             &= \Delta(w_k)\ \ .
\end{align*}
This also applies to the binary case. For some XNFA $N$, let $M_w = \prod_{i=1}^k M_{\sigma_i}$ represent a word $w = \sigma_1\sigma_2...\sigma_k$, where $M_{\sigma_i} = M_a$ if $\sigma_i = a$, and similarly for $b$. Now, let $N'$ be an XNFA whose transition matrices are $M'_a = A^{-1}M_aA$ and $M'_b = A^{-1}M_bA$ for some non-singular $A$. Then $w$ is represented by

\begin{align*}
M'_w &= \prod_{i=1}^k M'_{\sigma_i} \\
     &= M'_{\sigma_1} M'_{\sigma_2} \cdots M'_{\sigma_k} \\
     &= (A^{-1}M_{\sigma_1}A) (A^{-1}M_{\sigma_2}A) \cdots (A^{-1}M_{\sigma_k}A) \\
     &= A^{-1} M_{\sigma_1} M_{\sigma_2} \cdots M_{\sigma_k} A \\
     &= A^{-1} M_w A \ .
\end{align*}
And so the weight of any word $w_k$ on $N'$ is
\begin{align*}
\Delta'(w) &= v(Q'_0)M'_wv(F')^T \\
             &= v(Q_0)A (A^{-1}M_wA)A^{-1}v(F)^T \\
             &= v(Q_0)M_wv(F)^T \\
             &= \Delta(w) \ .
\end{align*}
Note that the above discussion does not rely on the fact that there are only two alphabet symbols, and so applies in general to the $r$-ary case as well.

\subsection{Self-verifying automata (SV-NFA)}
\label{subsec:sv-nfa}

Self-verifying NFA (SV-NFA)~\cite{Assent2007,Hromkovic1999,Jiraskova2011} are automata with two kinds of final states, namely accept states and reject states, as well as neutral non-final states. It is required that for any word, one or more of the paths for that word reach a single kind of final state, i.e. either accept states or reject states are reached, but not both. Consequently, self-verifying automata reject words explicitly if they reach a reject state, in contrast to NFA, where rejection is the result of a failure to reach an accept state.
\begin{definition}
\label{def:svnfa}
A  6-tuple 
$N = (Q, \Sigma, \delta, Q_0, F^a, F^r)$ 
is a self-verifying nondeterministic finite automaton (SV-NFA), where $Q,\Sigma, \delta$ and $Q_0$ are defined as for standard NFA. $F^a \subseteq Q$ and $F^r \subseteq Q$ are the sets of accept and reject states, respectively. The remaining states, that is, the states belonging to $Q \setminus (F^a \cup F^r)$, are called neutral states. For each input string $w$ in $\Sigma^*$, it is required that there exists at least one path ending in either an accept or a reject state; that is, $\delta (q_0,w) \cap (F^a \cup F^r) \neq \emptyset$ for any $q_0 \in Q_0$, and there are no strings $w$ such that both $\delta(q_0,w) \cap F^a$ and $\delta(q_1,w) \cap F^r$ are nonempty, for any $q_0,q_1 \in Q_0$.
\end{definition}
Since any SV-NFA either accepts or rejects any string $w \in \Sigma^*$ explicitly, its equivalent DFA must do so too. The path for each $w$ in a DFA is unique, so each state in the DFA is an accept or reject state. Hence, for any DFA state $d$, there is some SV-NFA state $q_i \in d$ such that $q_i \in F^a$ (and hence $d \in F^a_D$) or $q_i \in F^r$ (and hence $d \in F^r_D$). Since each state in the DFA is a subset of states of the SV-NFA, accept and reject states cannot occur together in a DFA state. That is, if $d$ is a DFA state, then for any $p,q \in d$, if $p \in F^a$ then $q \notin F^r$ and vice versa. We refer to the equivalent DFA of some SV-XNFA as its equivalent SV-XDFA to indicate that every state must accept or reject and that parity acceptance holds given the subset construction. Any SV-XDFA is equivalent to an XDFA, so SV-XNFA accept the class of regular languages.

\subsection{Self-verifying symmetric difference automata (SV-XNFA)}
\label{subsec:sv-xnfa}
In \cite{Marais2016}, self-verifying symmetric difference automata (SV-XNFA) were defined as a combination of the notions of symmetric difference automata and self-verifying automata, but only the unary case was examined. We now restate the definition of SV-XNFA in order to present results on larger alphabets in Section~\ref{sec:bin-svxnfa}. Note, however, that the definition is slightly amended: in \cite{Marais2016}, the implicit assumption was made that no SV-XNFA state could be both an accept state and a reject state. This assumption is explored in detail for the unary case in \cite{Marais2017}, but for our current purposes it suffices to say that such a requirement removes the equivalence between XNFA and weighted automata over GF(2), which is essential for certain operations on XNFA, such as minimisation \cite{VanderMerwe2012}. This implies that parity acceptance applies to SV-XNFA, where the condition for self-verification (SV-condition) is that for any word, an odd number of paths end in either accept states or reject states, but not both. In terms of the equivalent XDFA, this is equivalent to requiring that any XDFA state contain either an odd number of accept states or an odd number of rejects states, but not both. If an XNFA state is both an accept state and a reject state, it contributes to both counts.

\begin{definition}
\label{def:svxnfa}
A 6-tuple $N = (Q, \Sigma, \delta, Q_0, F^a, F^r)$
is a self-verifying symmetric difference finite automaton (SV-XNFA), where $Q,\Sigma, \delta$ and $Q_0$ are defined as for XNFA, and $F^a$ and $F^r$ are defined as for SV-NFA, except that $F^a \cap F^r$ need not be empty.   That is, each state in the SV-XDFA equivalent to $N$ must contain an odd number of states from either $F^a$ or $F^r$, but not both, and some SV-XNFA states may belong to both $F^a$ and $F^r$.
\end{definition}
The SV-condition for XNFA implies that if a state in the SV-XDFA of an SV-XNFA $N$ contains an odd number of states from $F^a$, it may also contain an even number of states from $F^r$, and hence belong to $F^a_D$, and vice versa. An SV-XDFA state may contain any number of neutral states from $N$.

The choice of $F^a$ and $F^r$ for a given SV-XNFA $N$ is called an \textit{SV-assignment} of $N$. An SV-assignment where either $F^a$ or $F^r$ is empty, is called a \textit{trivial SV-assignment}. Otherwise, if both $F^a$ and $F^r$ are nonempty, the SV-assignment is non-trivial.

\section{XNFA and linear feedback shift registers}
\label{sec:uxnfa-lfsr}
In \cite{VanZijl2001} it is shown that unary XNFA are equivalent to linear feedback shift registers (LFSRs). Specifically, a matrix $M$ with characteristic polynomial $c(X)$ is associated with a certain cycle structure of sets of XNFA states (or of XDFA states), and the choice of $Q_0$ determines which cycle represents the behaviour of a specific unary XNFA. The cycle structure is induced by $c(X)$, so any matrix that has $c(X)$ as its characteristic polynomial has the same cycle structure, although the states occurring in the cycles differ according to each specific matrix.

For the $r$-ary case, the transition matrix for each symbol is associated with its own cycle structure, and the choice of $Q_0$ determines which cycle is realised in the $r$-ary XNFA for each symbol. There are $2^n-1$ possible choices for $Q_0$ (we exclude the empty set). Evidently, the cycles associated with each symbol might overlap, and so the structure of the $r$-ary XNFA would not be cyclic itself, although the transitions for each symbol would exhibit cyclic behaviour. Specifically, for an $r$-ary XNFA $N$ and some symbol $\sigma \in \Sigma$, we refer to the cycle structure of $N$ on $\sigma$ as the cycle structure resulting from considering only transitions on $\sigma$. Our main results will be derived from examining the cycle structure induced by each symbol of the alphabet of the automaton, as well as the ways in which the cycles overlap.

For any $c(X) = X^n + c_{n-1}X^{n-1} + \cdots + c_1X + c_0$ there is a normal form matrix $M$ of the form given in Fig.~\ref{fig:companion-matrix}, such that $c(X) = \text{det}(XI-M)$, where $I$ is the identity matrix. We say that $M$ is in canonical form.

\begin{figure}
\begin{displaymath}
M=\left[\begin{array}{ccccc} 
0 & \ 1 & 0 & \cdots & 0 \\ 
0 & \ 0 & 1 & \cdots & 0 \\ 
\vdots & \ \vdots & & \vdots & \vdots \\
0 & \ 0 &  & \cdots  & 1 \\
c_0 & \ c_1 & \cdots & c_{n-2} & c_{n-1} 
\end{array}\right]\
\end{displaymath}
\caption{Companion matrix for $c(X) = X^n + c_{n-1}X^{n-1} + ... + c_1X + c_0$}
\label{fig:companion-matrix}
\end{figure}
In the next lemma, it will be convenient to represent XDFA states $d_s \subseteq Q$ as $s = \langle s_{n-1}, s_{n-2}, ... , s_1, s_0 \rangle$, where $s_i = 1$ if $q_i \in d_s$ and $0$ otherwise. The lemma is adapted from \cite{Stone1973} on the basis of the equivalence between unary XNFA and LFSRs.
\begin{lemma}
\label{lemma:gf2n-state-mapping}
Let $M_\sigma$ be a transition matrix representing transitions on $\sigma$ for some XNFA $N$, with characteristic polynomial $c_\sigma(X)$, and let $M_\sigma$ be in canonical form. Let $f$ be a bijection of the states of the equivalent XDFA $N_D$ onto polynomials of degree $n-1$, such that $f$ maps the state $s = \langle s_{n-1},s_{n-2},...,s_1,s_0 \rangle$ into the polynomial $f(s) = s_{n-1}X^{n-1} + s_{n-2}X^{n-2} + \cdots + s_1X + s_0$. Then $f$ maps the state $M_\sigma \cdot s$ into the polynomial $Xf(s) \text{ mod }c_\sigma(X)$.
\end{lemma} 
Lemma~\ref{lemma:gf2n-state-mapping} provides a mapping between polynomials over $GF(2)$ and the states of XDFA. The XDFA state arrived at after a transition from state $s$ on $\sigma$ corresponds to the polynomial which results from multiplying $f(s)$ by $X$ in the polynomial algebra of $GF(2)[X]$ modulo $c(X)$.
\begin{example}
\label{ex:binary-xnfa-gf2n}
Let $N$ be a binary XNFA (shown in Figure~\ref{ex:binary-xnfa-gf2n}), where $M_a$ is the normal form matrix of $c_a(X) = X^4 + X^2 + X + 1$ and $M_b$ is the normal form matrix of $c_b(X) = X^4 + X^3 + X + 1$. $M_a$ and $M_b$ are given in Fig.~\ref{fig:binary-ex-m-a} and \ref{fig:binary-ex-m-b}. The resulting XDFA is shown in Figure~\ref{fig:binary-xdfa-gf2n}, while some examples comparing state transitions and polynomial multiplication are shown in Table~\ref{tab:binary-xnfa-gf2n}. Note that, for now, the focus is on the cyclic behaviour of the equivalent XDFA, and so we do not refer to any final states.
\begin{figure}
\begin{adjustbox}{valign=t}
\begin{minipage}[t]{0.45\linewidth}
\centering
\begin{displaymath}
M_a=\left[\begin{array}{cccc} 
0 & 1 & 0 & 0 \\ 
0 & 0 & 1 & 0 \\ 
0 & 0 & 0 & 1 \\
1 & 1 & 1 & 0
\end{array}\right]\
\end{displaymath}
\caption{Example~\ref{ex:binary-change-basis}, transition matrix for $a$}
\label{fig:binary-ex-m-a}
\end{minipage}
\end{adjustbox}
\hspace{0.5cm}
\begin{adjustbox}{valign=t}
\begin{minipage}[t]{0.45\linewidth}
\centering
\begin{displaymath}
M_b=\left[\begin{array}{cccc} 
0 & 1 & 0 & 0 \\ 
0 & 0 & 1 & 0 \\ 
0 & 0 & 0 & 1 \\
1 & 1 & 0 & 1
\end{array}\right]\
\end{displaymath}
\caption{Example~\ref{ex:binary-change-basis}, transition matrix for $b$}
\label{fig:binary-ex-m-b}
\end{minipage}
\end{adjustbox}
\end{figure}
\begin{figure}
\begin{adjustbox}{valign=t}
\begin{minipage}[t]{0.45\linewidth}
\centering
\includegraphics[width=0.85\textwidth]{example1xnfa.tikz}
\caption{Example~\ref{ex:binary-xnfa-gf2n}, $N$}
\label{fig:binary-xnfa-gf2n}
\end{minipage}
\end{adjustbox}
\hspace{0.5cm}
\begin{adjustbox}{valign=t}
\begin{minipage}[t]{0.45\linewidth}
\centering
\includegraphics[width=0.75\textwidth]{example1xdfa.tikz}
\caption{Example~\ref{ex:binary-xnfa-gf2n}, $N_D$}
\label{fig:binary-xdfa-gf2n}
\end{minipage}
\end{adjustbox}
\end{figure}
\begin{table}
\begin{adjustbox}{valign=t}
\begin{minipage}[t]{\linewidth}
\centering
\caption{Transitions on $\delta$ correspond to multiplication by $X$}
\setlength{\tabcolsep}{5pt}
\renewcommand{\arraystretch}{1.1}
\begin{tabular}{l|ll}
\label{tab:binary-xnfa-gf2n}
$\delta_D(s,\sigma)$ & $Xf(s)\text{ mod }c_\sigma(X)$ \\
\hline
$\delta_D(\{q_0\},a) = \{q_1\}$                  & $X(1)$ & $= X$ \\
$\delta_D(\{q_3\},a) = \{q_0,q_1,q_2\}$          & $X(X^3)$ & $= X^4\text{ mod }c_a(X)$ \\
                                         && $= X^2 + X + 1$ \\
$\delta_D(\{q_0,q_2,q_3\},a) = \{q_0,q_2,q_3\}$  & $X(X^3 + X^2 + 1)$ & $= X^4 + X^3 + X\text{ mod }c_a(X)$ \\
                                         && $= X^3 + X^2 + 1$ \\
\hline
$\delta_D(\{q_1\},b) = \{q_2\}$                  & $X(X)$ & $= X^2$ \\
$\delta_D(\{q_0,q_1,q_3\},b) = \{q_0,q_2,q_3\}$  & $X(X^3 + X + 1)$ & $= X^4 + X^2 + X\text{ mod }c_b(X)$ \\
                                         && $= X^3 + X^2 + 1$ \\
$\delta_D(\{q_1,q_2,q_3\},b) = \{q_0,q_1,q_2\}$  & $X(X^3 + X^2 + X)$ & $= X^4 + X^3 + X^2\text{ mod }c_b(X)$ \\
                                         && $= X^2 + X + 1$ \\
\end{tabular}
\end{minipage}
\end{adjustbox}
\end{table}
\end{example}
\section{Non-unary SV-XNFA}
\label{sec:bin-svxnfa}
The upper bound on state complexity is simply $2^n-1$, since this is the number of non-empty subsets for any set of $n$ XNFA states. We now work towards establishing a lower bound on state complexity. First, we restate the following lemma from \cite{Marais2016} for the unary case.
\begin{lemma}
\label{lemma:cycle-with-odd-sized-states}
Let $c(X) = (X+1)\phi(X)$ be a polynomial of degree $n$ with non-singular normal form matrix $M$, and let $N$ be a unary XNFA with transition matrix $M$ and $Q_0 = \{q_0\}$. Then the equivalent XDFA $N_D$ has the following properties:
\begin{enumerate}
\item $|Q_D| > n$
\item $|d|$ is odd for $d \in Q_D$
\item $[q_0],[q_1],...,[q_{n-1}] \in Q_D$
\end{enumerate}
\end{lemma}
$|d|$ is the number of XNFA states in the XDFA state $d \subseteq Q$, or the number of one's in the representation of $d$ as $\langle s_{n-1},s_{n-2},...,s_1,s_0 \rangle$ where $s_i = 1$ if $q_i \in d$ and $0$ otherwise.
\begin{theorem}
\label{theorem:lower-bound-2n-1}
Let $M_{\sigma_1}$, $M_{\sigma_2}$, ..., $M_{\sigma_r}$ be the normal form matrices of $r$ polynomials $c_{\sigma_1}(X) = (X+1)\phi_{\sigma_1}(X)$, $c_{\sigma_2}(X) = (X+1)\phi_{\sigma_2}(X)$, ..., $c_{\sigma_r}(X) = (X+1)\phi_{\sigma_r}(X)$, respectively, and let $M_{\sigma_1}$, $M_{\sigma_2}$, ..., $M_{\sigma_r}$ be the transition matrices of some $r$-ary XNFA $N$ with $\Sigma = \{\sigma_1,\sigma_2,...,\sigma_r\}$ and $Q_0 = \{q_0\}$. Then the number of states in the equivalent XDFA $N_D$ does not exceed $2^{n-1}$. Furthermore, any choice of $F^a$ and $F^r$ such that $F^a \cup F^r = Q$ and $F^a \cap F^r = \emptyset$ is an SV-assignment.
\end{theorem}
\begin{proof}
By Lemma~\ref{lemma:cycle-with-odd-sized-states}, $|d|$ is odd for $d \in Q_D$ in the unary case. That is, for any symbol with a transition matrix whose polynomial has $X+1$ as a factor, a transition from an odd-sized XDFA state is to another odd-sized XDFA state. Since $Q_0 = \{q_0\}$ and $|\{q_0\}|$ is odd, and $c_{\sigma_1}(X)$, $c_{\sigma_2}(X)$,...$c_{\sigma_r}(X)$, have $X+1$ as a factor, only odd-sized states are reachable on any transition. The number of XDFA states $d$ such that $|d|$ is odd is $2^n/2 = 2^{n-1}$, and so $N_D$ can have at most $2^{n-1}$ states.
Since every XDFA state contains an odd number of XNFA states, any choice of $F^a$ and $F^r$ such that $F^a \cup F^r = Q$ and $F^a \cap F^r = \emptyset$ is an SV-assignment.
\end{proof}
The following lemma provides further information on the cycle structure induced by polynomials with $X+1$ as a factor.
\begin{lemma}
\label{lemma:phi-1-cycle}
Let $c_\sigma(X) = (X+1)\phi(X)$. Then, in the normal form matrix $M_\sigma$ of $c_\sigma(X)$, which is the transition matrix on some symbol $\sigma$ for an XNFA, the state mapped to $\phi(X)$ as described in Lemma~\ref{lemma:gf2n-state-mapping}, i.e. $d_\phi$, is contained in a cycle of length one, when considering only transitions on $\sigma$.
\end{lemma}
\begin{proof}
Consider the following:
\begin{align*}
(X+1)\phi(X) &= c_\sigma(X) \\
X\phi(X) + \phi(X) &= c_\sigma(X) \\
X\phi(X) &= \phi(X) + c_\sigma(X)
\end{align*}
Therefore, $X\phi(X) = \phi(X)$ in the representation of $GF(2^n)$ as polynomials over GF(2) modulo $c_\sigma(X)$. By Lemma~\ref{lemma:gf2n-state-mapping}, this corresponds to $\delta_D(d_\phi,\sigma) = d_\phi$.
\end{proof}
We now present a witness language for any $n$ to show that $2^{n-1}$ is a lower bound on the state complexity of SV-XNFA with non-unary alphabets. First, we restate the following theorem from \cite{Marais2016}.
\begin{theorem}
\label{theorem:sv-xnfa-2n-1-1-complexity}
For any $n \geq 2$, there is an SV-XNFA $N$ whose equivalent $N_D$ has $2^{n-1}-1$ states.
\end{theorem}
\begin{lemma}
\label{lemma:witness-machine}
Let $\phi(X) = X^{n-1} + \phi_{n-2}X^{n-2} + \cdots + \phi_1X + \phi_0$ be any primitive polynomial of degree $n-1$. Let $N$ be a binary XNFA, and let the transition matrix on $a$ be the normal form matrix of $c_a(X) = (X+1)\phi(X)$ and the transition matrix on $b$ be the normal form matrix of $c_b(X) = X^n + \phi(X)$. Then the equivalent XDFA of the XNFA with $Q_0 = \{q_0\}$ contains exactly $2^{n-1}$ odd-sized states.
\end{lemma}
\begin{proof}
We write $c_a(X)$ and $c_b(X)$ in the following way:
\begin{align*}
c_a(X) &= X^n + c_{n-1}X^{n-1} + \cdots + c_{1}X + c_{0} \\
c_b(X) &= X^n + \phi_{n-1}X^{n-1} + \phi_{n-2}X^{n-2} + \cdots + \phi_1X + \phi_0
\end{align*}
Since $\phi(X)$ is primitive, it has no roots in GF(2), including 1, so it must have an odd number of non-zero terms. Therefore, by Lemma~\ref{lemma:gf2n-state-mapping}, $|d_\phi|$ is odd. Furthermore, $c_b(X)$ has an even number of non-zero terms, and so has 1 as a root. Consequently, $c_b(X)$ has $X+1$ as a factor.

The transition matrices $M_a$ and $M_b$ are given in Fig.~\ref{fig:binary-theorem-m-a} and \ref{fig:binary-theorem-m-b}. Note that they are both non-singular. Let $Q_0 = \{q_0\}$. Then by Theorem~\ref{theorem:sv-xnfa-2n-1-1-complexity}, the cycle structure on $a$ is equivalent to an XDFA cycle with $2^{n-1}-1$ states, all of which, by Lemma~\ref{lemma:cycle-with-odd-sized-states}, have odd size. Also, by Lemma~\ref{lemma:phi-1-cycle}, $d_\phi$ is not contained in this cycle. This means that on $a$, every odd-sized state in the XDFA is reached except for $d_\phi$.
\begin{figure}
\begin{adjustbox}{valign=t}
\begin{minipage}[t]{0.45\linewidth}
\centering
\begin{displaymath}
M_a=\left[\begin{array}{ccccc} 
0 & \ 1 & 0 & \cdots & 0 \\ 
0 & \ 0 & 1 & \cdots & 0 \\ 
\vdots & \ \vdots & & \vdots & \vdots \\
0 & \ 0 &  & \cdots  & 1 \\
c_0 & \ c_1 & \cdots & c_{n-2} & c_{n-1} 
\end{array}\right]\
\end{displaymath}
\caption{Lemma~\ref{lemma:witness-machine}, transition matrix for $a$}
\label{fig:binary-theorem-m-a}
\end{minipage}
\end{adjustbox}
\hspace{0.5cm}
\begin{adjustbox}{valign=t}
\begin{minipage}[t]{0.45\linewidth}
\centering
\begin{displaymath}
M_b=\left[\begin{array}{ccccc} 
0 & \ 1 & 0 & \cdots & 0 \\ 
0 & \ 0 & 1 & \cdots & 0 \\ 
\vdots & \ \vdots & & \vdots & \vdots \\
0 & \ 0 &  & \cdots  & 1 \\
\phi_0 & \ \phi_1 & \cdots & \phi_{n-2} & \phi_{n-1} 
\end{array}\right]\
\end{displaymath}
\caption{Lemma~\ref{lemma:witness-machine}, transition matrix for $b$}
\label{fig:binary-theorem-m-b}
\end{minipage}
\end{adjustbox}
\end{figure}
Now, from $M_b$ it follows directly that $\delta_D(\{q_{n-1}\},b) = d_\phi$. Furthermore, since $X+1$ is a factor of $c_b$, every transition from an odd-sized state on $b$ is to an odd-sized state. Consequently, the binary XNFA $N$ is equivalent to an XDFA that reaches all $2^{n-1}$ odd-sized states and none other.
\end{proof}
\begin{theorem}
\label{theorem:witness-binary}
For any $n \geq 2$, there is a language $\mathcal{L}_n$ so that some $n$-state binary SV-XNFA accepts $\mathcal{L}_n$ and the minimal SV-XDFA that accepts $\mathcal{L}_n$ has $2^{n-1}$ states.
\end{theorem}
\begin{proof}
Let $c_a(X) = (X+1)\phi(X)$ and $c_b = X^n + \phi(X)$, where $\phi(X)$ is a primitive polynomial and let $c_a(X)$ and $c_b(X)$ have degree $n$. We construct an SV-XNFA $N$ with $n$ states whose equivalent XDFA $N_D$ has $2^{n-1}$ states as in Lemma~\ref{lemma:witness-machine}, and let $F^a = \{q_0\}$ and $F^r = Q \setminus F^a$. Recall that for $N$, we have $\delta: Q \times \Sigma \rightarrow 2^Q$, and for $N_D$, we have $\delta_D: 2^Q \times \Sigma \rightarrow 2^Q$.

Let $\mathcal{L}_n^1 = a^{(2^{n-1}-1)i + j}$ for $i \geq 0$ and $j \in J$, where $J$ is some set of integers, represent a subset of the language accepted by $N$ that consists only of strings containing $a$. Now, from the transition matrix of $N$ it follows that $0,n \in J$, while $1,2,...,n-1 \notin J$, since $q_0 \in \delta(q_0,a^n)$, but $q_0 \notin \delta(q_0,a^m)$ for $m < n$.

If there is an $N_D'$ with fewer than $2^{n-1}-1$ states that accepts $\mathcal{L}_n^1$, then there must be some $d_j \in Q_D$ such that $\{q_0\} \subset d_j$, $q_0 \in \delta_D(d_j,a^n)$ and there is no $m < n$ so that $q_0 \in \delta_D(d_j,a^m)$. That is, if $N_D'$ exists, then on $N_D$, $\delta_D(\{q_0\},a) = \delta_D(d_j,a)$, and $\delta_D(\{q_1\},a) = \delta_D(d_{j+1},a)$ etc.

Let $d_k$ be any state in $N_D$ such that $d_k \neq \{q_0\}$. Let $max(d_k)$ be the largest subscript of any SV-XNFA state in $d_k$. Then $max(d_k) > 0$. Let $m = n - max(d_k)$, so $m < n$. Then, from the transition matrix of $N$, it follows that $q_0 \in \delta_D(d_k,a^m)$. That is, for any $d_k$ there is an $m < n$ so that $q_0 \in \delta_D(d_k,a^m)$. Therefore, there is no $N_D'$ with fewer than $2^{n-1}-1$ states that accepts $\mathcal{L}_n^1$.

Now, let $\mathcal{L}_n^2 = b^{n}a^*$, which is also a subset of the language accepted by $N$. In order to accept this language, after reading $b^n$, a state must have been reached whereafter every transition on $a$ must result in an accept state, i.e. an XDFA state containing $q_0$. But there is only one such state, and that is $d_\phi$, since $\delta_D(d_\phi,a) = d_\phi$, which is excluded from the cycle needed to accept $\mathcal{L}_n^1$. Therefore, all $2^{n-1}$ odd-sized states are necessary to accept $\mathcal{L}^1 \cup \mathcal{L}^2$. Let $\mathcal{L}_n$ be the language accepted by $N$, then since $\mathcal{L}_n^1 \cup \mathcal{L}_n^2 \subset \mathcal{L}_n$, at least $2^{n-1}$ states are necessary to accept $\mathcal{L}_n$.
\end{proof}
We illustrate Theorem~\ref{theorem:witness-binary} for $n=4$.

\begin{example}
\label{ex:lower-bound}
Let $\phi(X) = X^3 + X + 1$, which is a primitive polynomial. Now, let $N$ be an XNFA with transition matrices $M_a$ and $M_b$. $M_a$ is the normal form matrix of $c_a(X) = (X+1)\phi(X) = X^4 + X^3 + X^2 + 1$ and $M_b$ the normal form matrix of $c_b(X) = X^4 + \phi(X) = X^4 + X^3 + X + 1$. Let $Q_0 =\{q_0\}$ and let $F^a = \{q_0\}$ and $F^r = \{q_1,q_2,q_3\}$. $M_a$ and $M_b$ are shown in Figures~\ref{fig:lower-bound-example-m-a} and \ref{fig:lower-bound-example-m-b}, while $N$ and its equivalent XDFA $N_D$ are shown in Figures~\ref{fig:lower-bound-xnfa} and \ref{fig:lower-bound-xdfa}. We have $\mathcal{L}^1 = a^{7i + j}$ for $i \geq 0$ and $j \in \{0,4,5\}$ and $\mathcal{L}^2 = bbbba^*$.
\begin{figure}
\begin{adjustbox}{valign=t}
\begin{minipage}[t]{0.45\linewidth}
\centering
\begin{displaymath}
M_a=\left[\begin{array}{ccccc} 
0 & 1 & 0 & 0 \\ 
0 & 0 & 1 & 0 \\ 
0 & 0 & 0 & 1 \\
1 & 0 & 1 & 1
\end{array}\right]\
\end{displaymath}
\caption{Example~\ref{ex:lower-bound}, transition matrix for $a$}
\label{fig:lower-bound-example-m-a}
\end{minipage}
\end{adjustbox}
\hspace{0.5cm}
\begin{adjustbox}{valign=t}
\begin{minipage}[t]{0.45\linewidth}
\centering
\begin{displaymath}
M_b=\left[\begin{array}{ccccc} 
0 & 1 & 0 & 0 \\ 
0 & 0 & 1 & 0 \\ 
0 & 0 & 0 & 1 \\
1 & 1 & 0 & 1
\end{array}\right]\
\end{displaymath}
\caption{Example~\ref{ex:lower-bound}, transition matrix for $b$}
\label{fig:lower-bound-example-m-b}
\end{minipage}
\end{adjustbox}
\end{figure}
%

\begin{figure}
\begin{adjustbox}{valign=t}
\begin{minipage}[t]{0.45\linewidth}
\centering
\includegraphics[width=0.85\textwidth]{example2xnfa.tikz}
\caption{Example~\ref{ex:lower-bound}, $N$}
\label{fig:lower-bound-xnfa}
\end{minipage}
\end{adjustbox}
\hspace{0.5cm}
\begin{adjustbox}{valign=t}
\begin{minipage}[t]{0.45\linewidth}
\centering
\includegraphics[width=0.75\textwidth]{example2xdfa.tikz}
\caption{Example~\ref{ex:lower-bound}, $N_D$}
\label{fig:lower-bound-xdfa}
\end{minipage}
\end{adjustbox}
\end{figure}
\end{example}
The following is a simple corollary of Theorem~\ref{theorem:witness-binary}.
\begin{corollary}
\label{theorem:witness-n-ary}
For any $m,n \geq 2$, there is a language $\mathcal{L}'_n$ so that some $n$-state $m$-ary SV-XNFA accepts $\mathcal{L}'_n$ and the minimal SV-XDFA that accepts $\mathcal{L}'_n$ has $2^{n-1}$ states.
\end{corollary}

We now show that any given SV-XNFA can be used to obtain another one via a so-called change of basis.
\begin{theorem}
Given any SV-XNFA $N = (Q, \Sigma, \delta, Q_0, F^a, F^r)$ with $n$ states and transition matrices $M_{\sigma_1}$, $M_{\sigma_2}$, ..., $M_{\sigma_r}$, and any non-singular $n \times n$ matrix $A$, we encode $Q_0$ as a vector $v(Q_0)$ of length $n$ over GF(2) and $F^a$ and $F^r$ as vectors $v(F^a)$ and $v(F^r)$ respectively. Then there is an SV-XNFA $N' = (Q, \Sigma, \delta', Q_0', F'^a, F'^r)$ where $M'_{\sigma_i} = A^{-1}M_{\sigma_i}A$ for $0 \leq i \leq r$, $v(Q'_0) = v(Q_0)A$, $v(F'^a)^T = A^{-1}v(F^a)^T$ and $v(F'^r)^T = A^{-1}v(F^r)^T$, and $N'$ accepts the same language as $N$.
\end{theorem}
\begin{proof}
In the discussion in Section~\ref{subsec:xnfa} we showed that for XNFA, the change of basis described on an XNFA $N$ that results in $N'$, $\Delta'(w) = \Delta(w)$. We extend this to SV-XNFA by defining two new functions. Recall that $M_w$ represents the sequence of matrix multiplications for some $w$ of length $k$, and that $M'_w = A^{-1}M_wA$. Then, let
\begin{align*}
accept(w) &= v(Q_0)M_w v(F^a)^T \\
reject(w) &= v(Q_0)M_w v(F^r)^T\ \ .
\end{align*}
The SV-condition is that $accept(w) \neq reject(w)$ for any $w \in \Sigma^*$. Similar to $\Delta(w)$, we have
\begin{align*}
accept'(w) &= v(Q'_0)M'_w v(F'^a)^T \\
           &= v(Q_0)A (A^{-1} M_w A) A^{-1}v(F^a) \\
           &= v(Q_0) M_w v(F^a) \\
           &= accept(w)
\end{align*}
and
\begin{align*}
reject'(w) &= v(Q'_0)M'_w v(F'^r)^T \\
           &= v(Q_0)A (A^{-1} M_w A) A^{-1}v(F^r) \\
           &= v(Q_0) M_w v(F^r) \\
           &= reject(w)
\end{align*}
Clearly, the SV-condition is met by $accept'$ and $reject'$, and so $N'$ is an SV-XNFA that accepts the same language as $N$.
\end{proof}
The number of non-singular $n \times n$ matrices over GF(2) (including the identity matrix) is $|GL(n,\mathbb{Z}_2)| = \prod_{k=0}^{n-1} (2^n-2^k)$, and so, up to isomorphism, for any SV-XNFA at most another $|GL(n,\mathbb{Z}_2)|-1$ equivalent SV-XNFA can be found.
\begin{example}
\label{ex:binary-change-basis}
Let $N$ be an SV-XNFA with
alphabet $\Sigma = \{a,b,c\}$, and the following transition matrices: $M_a$ is the normal form matrix of $c(X) = X^4 + X^3 + X^2 + 1$, $M_b$ is the normal form matrix of $X^4 + X^3 + X + 1$, and $M_c$ is the normal form matrix of $c(X) = X^4 + X^2 + X + 1$. Let $Q_0 = \{q_0\}$,  $F^a = \{q_0,q_2\}$ and $F^r = \{q_1,q_3\}$. Figure~\ref{fig:trinary-xnfa-normal} shows $N$ and the equivalent XDFA $N_D$ is given in Figure~\ref{fig:binary-xdfa-normal}, where a double edge indicates an accept state and a thick edge indicates a reject state. Consider the following matrix $A$:
\begin{displaymath}
A=\left[\begin{array}{cccc} 
0 & 1 & 1 & 1 \\ 
1 & 0 & 1 & 0 \\ 
1 & 1 & 0 & 0 \\
0 & 1 & 0 & 1
\end{array}\right]\ \ \ .
\end{displaymath}
We use $A$ to make a change of basis from $N$ to $N'$. Let $N'$ be an XNFA with $\Sigma=\{a,b,c\}$, where $M'_a = A^{-1}M_a A$, $M'_b = A^{-1}M_b A$ and $M'_c = A^{-1}M_c A$. Furthermore, let $v(Q'_0) = v(Q_0)A$, i.e. $Q'_0 = \{q_1,q_2,q_3\}$. Finally, let $v(F'^a)^T = A^{-1}v(F^a)^T$ and $v(F'^r)^T = A^{-1}v(F^r)^T$, i.e. $F'^a = \{q_0,q_2\}$ and $F'^r = \{q_2,q_3\}$. Figure~\ref{fig:trinary-xnfa-change}, shows $N'$, with a double edge indicating an accept state, a thick edge indicating a reject state and a thick double edge indicating a state that is both an accept state and a reject state. Figure~\ref{fig:binary-xdfa-change} gives the equivalent XDFA $N'_D$. It is worth noting that, although $N'$ has a different structure than $N$, $N'_D$ has the same structure as $N_D$, and accepts the same language. Also, note that in $N'_D$, the state $\{q_0,q_1,q_2\}$ is a reject state, because it contains an even number of accept states, namely $q_0$ and $q_2$, but an odd number of reject states, namely $q_2$.
\end{example}
\begin{figure}
\begin{adjustbox}{valign=t}
\begin{minipage}[t]{0.45\linewidth}
\centering
\includegraphics[width=0.85\textwidth]{example20xnfa.tikz}
\caption{Example~\ref{ex:binary-change-basis}, $N$}
\label{fig:trinary-xnfa-normal}
\end{minipage}
\end{adjustbox}
\hspace{0.5cm}
\begin{adjustbox}{valign=t}
\begin{minipage}[t]{0.45\linewidth}
\centering
\includegraphics[width=0.75\textwidth]{example20xdfa.tikz}
\caption{Example~\ref{ex:binary-change-basis}, $N_D$}
\label{fig:binary-xdfa-normal}
\end{minipage}
\end{adjustbox}
\end{figure}
\begin{figure}
\begin{adjustbox}{valign=t}
\begin{minipage}[t]{0.45\linewidth}
\centering
\includegraphics[width=\textwidth]{example21xnfa.tikz}
\caption{Example~\ref{ex:binary-change-basis}, $N'$}
\label{fig:trinary-xnfa-change}
\end{minipage}
\end{adjustbox}
\hspace{0.5cm}
\begin{adjustbox}{valign=t}
\begin{minipage}[t]{0.45\linewidth}
\centering
\includegraphics[width=0.75\textwidth]{example21xdfa.tikz}
\caption{Example~\ref{ex:binary-change-basis}, $N'_D$}
\label{fig:binary-xdfa-change}
\end{minipage}
\end{adjustbox}
\end{figure}
%
%
\section{Conclusion}
\label{sec:conc}
We have given an upper bound of $2^n-1$ on the state complexity of SV-XNFA with alphabets larger than one, and a lower bound of $2^{n-1}$. We have also shown that, given any SV-XNFA with $n$ states, it is possible, up to isomorphism, to find at most another $|GL(n,\mathbb{Z}_2)|-1$ equivalent SV-XNFA via a change of basis.
\bibliographystyle{eptcs}
\bibliography{svxnfa}

\begin{thebibliography}{10}
\providecommand{\bibitemdeclare}[2]{}
\providecommand{\surnamestart}{}
\providecommand{\surnameend}{}
\providecommand{\urlprefix}{Available at }
\providecommand{\url}[1]{\texttt{#1}}
\providecommand{\href}[2]{\texttt{#2}}
\providecommand{\urlalt}[2]{\href{#1}{#2}}
\providecommand{\doi}[1]{doi:\urlalt{http://dx.doi.org/#1}{#1}}
\providecommand{\bibinfo}[2]{#2}

\bibitemdeclare{article}{Assent2007}
\bibitem{Assent2007}
\bibinfo{author}{Ira \surnamestart Assent\surnameend} \&
  \bibinfo{author}{Sebastian \surnamestart Seibert\surnameend}
  (\bibinfo{year}{2007}): \emph{\bibinfo{title}{An upper bound for transforming
  self-verifying automata into deterministic ones}}.
\newblock {\sl \bibinfo{journal}{RAIRO-Theoretical Informatics and
  Applications-Informatique Th{\'e}orique et Applications}}
  \bibinfo{volume}{41}(\bibinfo{number}{3}), pp. \bibinfo{pages}{261--265},
\doi{10.1051/ita:2007017}.

\bibitemdeclare{book}{Hopcroft1990}
\bibitem{Hopcroft1990}
\bibinfo{author}{John~E. \surnamestart Hopcroft\surnameend} \&
  \bibinfo{author}{Jeffrey~D. \surnamestart Ullman\surnameend}
  (\bibinfo{year}{1990}): \emph{\bibinfo{title}{Introduction to Automata
  Theory, Languages, and Computation}}, \bibinfo{edition}{1st} edition.
\newblock \bibinfo{publisher}{Addison-Wesley Longman Publishing Co., Inc.},
  \bibinfo{address}{Boston, MA, USA}.

\bibitemdeclare{inproceedings}{Hromkovic1999}
\bibitem{Hromkovic1999}
\bibinfo{author}{Juraj \surnamestart Hromkovic\surnameend} \&
  \bibinfo{author}{Georg \surnamestart Schnitger\surnameend}
  (\bibinfo{year}{1999}): \emph{\bibinfo{title}{On the Power of Las Vegas {II.}
  Two-Way Finite Automata}}.
\newblock In \bibinfo{editor}{Jir{\'{\i}} \surnamestart Wiedermann\surnameend},
  \bibinfo{editor}{Peter \surnamestart van Emde~Boas\surnameend} \&
  \bibinfo{editor}{Mogens \surnamestart Nielsen\surnameend}, editors: {\sl
  \bibinfo{booktitle}{Automata, Languages and Programming, 26th International
  Colloquium, ICALP'99, Prague, Czech Republic, July 11-15, 1999,
  Proceedings}}, {\sl \bibinfo{series}{Lecture Notes in Computer Science}}
  \bibinfo{volume}{1644}, \bibinfo{publisher}{Springer}, pp.
  \bibinfo{pages}{433--442}, \doi{10.1007/3-540-48523-6\_40}.


\bibitemdeclare{article}{Jiraskova2011}
\bibitem{Jiraskova2011}
\bibinfo{author}{Galina \surnamestart Jir{\'a}skov{\'a}\surnameend} \&
  \bibinfo{author}{Giovanni \surnamestart Pighizzini\surnameend}
  (\bibinfo{year}{2011}): \emph{\bibinfo{title}{Optimal simulation of
  self-verifying automata by deterministic automata}}.
\newblock {\sl \bibinfo{journal}{Information and Computation}}
  \bibinfo{volume}{209}(\bibinfo{number}{3}), pp. \bibinfo{pages}{528 -- 535}.
\newblock \bibinfo{note}{Special Issue: 3rd International Conference on
  Language and Automata Theory and Applications (LATA 2009)},
\doi{10.1016/j.ic.2010.11.017}.

\bibitemdeclare{unpublished}{Marais2017}
\bibitem{Marais2017}
\bibinfo{author}{Laurette \surnamestart Marais\surnameend} \&
  \bibinfo{author}{Lynette \surnamestart Van~Zijl\surnameend}:
  \emph{\bibinfo{title}{State Complexity of {U}nary {SV-XNFA} with {D}ifferent
  {A}cceptance {C}onditions}}.
\newblock \bibinfo{note}{Submitted for publication}.

\bibitemdeclare{inproceedings}{Marais2016}
\bibitem{Marais2016}
\bibinfo{author}{Laurette \surnamestart Marais\surnameend} \&
  \bibinfo{author}{Lynette \surnamestart van Zijl\surnameend}
  (\bibinfo{year}{2016}): \emph{\bibinfo{title}{Unary Self-verifying Symmetric
  Difference Automata}}.
\newblock In \bibinfo{editor}{Cezar \surnamestart C{\^{a}}mpeanu\surnameend},
  \bibinfo{editor}{Florin \surnamestart Manea\surnameend} \&
  \bibinfo{editor}{Jeffrey \surnamestart Shallit\surnameend}, editors: {\sl
  \bibinfo{booktitle}{Descriptional Complexity of Formal Systems - 18th {IFIP}
  {WG} 1.2 International Conference, {DCFS} 2016, Bucharest, Romania, July 5-8,
  2016. Proceedings}}, {\sl \bibinfo{series}{Lecture Notes in Computer
  Science}} \bibinfo{volume}{9777}, \bibinfo{publisher}{Springer}, pp.
  \bibinfo{pages}{180--191}, \doi{10.1007/978-3-319-41114-9_14}.

\bibitemdeclare{inproceedings}{VanderMerwe2012}
\bibitem{VanderMerwe2012}
\bibinfo{author}{Brink \surnamestart Van~der Merwe\surnameend},
  \bibinfo{author}{Hellis \surnamestart Tamm\surnameend} \&
  \bibinfo{author}{Lynette \surnamestart Van~Zijl\surnameend}
  (\bibinfo{year}{2012}): \emph{\bibinfo{title}{Minimal {DFA} for Symmetric
  Difference {NFA}}}.
\newblock In \bibinfo{editor}{Martin \surnamestart Kutrib\surnameend},
  \bibinfo{editor}{Nelma \surnamestart Moreira\surnameend} \&
  \bibinfo{editor}{Rog{\'{e}}rio \surnamestart Reis\surnameend}, editors: {\sl
  \bibinfo{booktitle}{Descriptional Complexity of Formal Systems - 14th
  International Workshop, {DCFS} 2012, Braga, Portugal, July 23-25, 2012.
  Proceedings}}, {\sl \bibinfo{series}{Lecture Notes in Computer Science}}
  \bibinfo{volume}{7386}, \bibinfo{publisher}{Springer}, pp.
  \bibinfo{pages}{307--318}, \doi{10.1007/978-3-642-31623-4_24}.

\bibitemdeclare{book}{Stone1973}
\bibitem{Stone1973}
\bibinfo{author}{Harold~S \surnamestart Stone\surnameend}
  (\bibinfo{year}{1973}): \emph{\bibinfo{title}{Discrete Mathematical
  Structures and their Applications}}.
\newblock \bibinfo{publisher}{Science Research Associates Chicago}.

\bibitemdeclare{inproceedings}{Vuillemin2009}
\bibitem{Vuillemin2009}
\bibinfo{author}{Jean \surnamestart Vuillemin\surnameend} \&
  \bibinfo{author}{Nicolas \surnamestart Gama\surnameend}
  (\bibinfo{year}{2009}): \emph{\bibinfo{title}{Compact Normal Form for Regular
  Languages as Xor Automata}}.
\newblock In \bibinfo{editor}{Sebastian \surnamestart Maneth\surnameend},
  editor: {\sl \bibinfo{booktitle}{Implementation and Application of Automata,
  14th International Conference, {CIAA} 2009, Sydney, Australia, July 14-17,
  2009. Proceedings}}, {\sl \bibinfo{series}{Lecture Notes in Computer
  Science}} \bibinfo{volume}{5642}, \bibinfo{publisher}{Springer}, pp.
  \bibinfo{pages}{24--33}, \doi{10.1007/978-3-642-02979-0_6}.

\bibitemdeclare{article}{VanZijl2004}
\bibitem{VanZijl2004}
\bibinfo{author}{Lynette \surnamestart van Zijl\surnameend}
  (\bibinfo{year}{2005}): \emph{\bibinfo{title}{Magic numbers for symmetric
  difference {NFAS}}}.
\newblock {\sl \bibinfo{journal}{Int. J. Found. Comput. Sci.}}
  \bibinfo{volume}{16}(\bibinfo{number}{5}), pp. \bibinfo{pages}{1027--1038},
  \doi{10.1142/S0129054105003455}.

\bibitemdeclare{inproceedings}{VanZijl2001}
\bibitem{VanZijl2001}
\bibinfo{author}{Lynette \surnamestart van Zijl\surnameend},
  \bibinfo{author}{John{-}Paul \surnamestart Harper\surnameend} \&
  \bibinfo{author}{Frank \surnamestart Olivier\surnameend}
  (\bibinfo{year}{2000}): \emph{\bibinfo{title}{The MERLin Environment Applied
  to *-NFAs}}.
\newblock In \bibinfo{editor}{Sheng \surnamestart Yu\surnameend} \&
  \bibinfo{editor}{Andrei \surnamestart Paun\surnameend}, editors: {\sl
  \bibinfo{booktitle}{Implementation and Application of Automata, 5th
  International Conference, {CIAA} 2000, London, Ontario, Canada, July 24-25,
  2000, Revised Papers}}, {\sl \bibinfo{series}{Lecture Notes in Computer
  Science}} \bibinfo{volume}{2088}, \bibinfo{publisher}{Springer}, pp.
  \bibinfo{pages}{318--326}, \doi{10.1007/3-540-44674-5_28}.

\end{thebibliography}

\end{document}